%% file: main.tex
\documentclass[11pt]{article}
\usepackage{fullpage}
\usepackage{graphicx} %

\usepackage{amsmath,amsthm,amssymb}
\usepackage{thmtools}
\usepackage[dvipsnames]{xcolor}
\usepackage[colorlinks,citecolor=blue,linkcolor=BrickRed]{hyperref}
\usepackage{fullpage}
\usepackage[algo2e,linesnumbered,ruled,vlined,algosection]{algorithm2e}
\usepackage[capitalize,noabbrev]{cleveref}
\usepackage{comment}
\SetKw{Continue}{continue}
\SetKw{Break}{break}
\SetKwInput{KwInput}{Input}
\SetKwInput{KwOutput}{Output}
\SetKwInput{KwParameter}{Parameter}

\newcommand{\defeq}{\stackrel{\mathrm{def}}{=}}
\newcommand{\diag}{\operatorname{diag}}
\newcommand{\R}{\mathbb{R}}
\newcommand{\argmin}{\operatorname*{argmin}}
\newcommand{\norm}[1]{\lVert#1\rVert}
\newcommand{\nnz}{\mathrm{nnz}}
\newcommand{\eps}{\varepsilon}

\newcommand{\poly}{\mathrm{poly}}

\declaretheorem[numberwithin=section,refname={Theorem,Theorems},Refname={Theorem,Theorems}]{theorem}

\declaretheorem[numberlike=theorem]{lemma}

\declaretheorem[numberlike=theorem]{claim}

\declaretheorem[numberlike=theorem,refname={Fact,Facts},Refname={Fact,Facts},name={Fact}]{fact}

\newcommand{\A}{\mathcal{A}}
\newcommand{\T}{\mathcal{T}}
\newcommand{\D}{\mathcal{D}}

\newcommand{\SDD}{\mathrm{SDD}}

\title{Near-Optimal Parallel $M$-Matrix Scaling}

\author{Chih-Hung Liu \\ National Taiwan University \\ \texttt{chliu@ntu.edu.tw} \and Aaron Sidford \\ Stanford University \\ \texttt{sidford@stanford.edu} \and Ta-Wei Tu \\ Stanford University \\  \texttt{taweitu@stanford.edu}}
\date{}

\begin{document}

\maketitle

\begin{abstract}
In this paper we consider the problem of computing a positive diagonal scaling of a $d \times d$ symmetric $M$-matrix so that all row and column sums are $1$. We provide a poly-logarithmic iteration algorithm that solves this problem to high accuracy, where the cost of each iteration is dominated by the cost of solving a Laplacian system. This improves upon the recent $\tilde{O}(d^{1/3})$ iteration bound of [Vladu STOC 2025] and yields a randomized algorithm for the problem that runs in near-linear time and poly-logarithmic parallel depth.

\end{abstract}

\input{intro}
\input{overview}
\input{scaling}

\input{ack}

\bibliographystyle{alpha}
\bibliography{reference}

\appendix

\input{flow}
\input{omitted_proof}

\end{document}

%% file: intro.tex
\section{Introduction}

We consider the fundamental \emph{symmetric matrix-scaling} problem where for an input symmetric matrix $A \in \R^{d \times d}$ the goal is to compute a diagonal matrix $X$ with positive diagonal entries such that $X A X$ has all row sums (and therefore column sums) equal to $1$. This problem encompasses the well-studied asymmetric variant of the problem where for an input matrix $B \in \R^{d \times d}$ the goal is to compute diagonals $L \in \R^{d \times d}$ and $R \in \R^{d \times d}$ with positive diagonal entries so that $L B R$ has row and column sums equal to $1$.
This is done simply by solving the symmetric matrix scaling problem for the anti-block-diagonal matrix with anti-diagonals $B$ and $B^\top$.

These matrix scaling problems are well studied and have numerous applications spanning many scientific disciplines~\cite{kruithof1937telefoonverkeersrekening,deming1940least,brown1959note,friedlander1961technique,herman1976iterative,raghavan1984pairs,rothblum1989scalings,LinialSW98,BalakrishnanHT04,idel2016review}, and notably it is known that linear programming can be reduced to them~\cite{KhachiyanK92}.
Matrix scaling is also sometimes used in practice to precondition linear systems, in hopes of improving the performance of the solvers~\cite{LivneG04,PockC11}.
There is a particularly long line of research on obtaining efficient algorithms for the special case where $A$ is \emph{nonnegative}~\cite{kalantari1996complexity,LinialSW98,BalakrishnanHT04,RoteZ07,CohenMTV17,Allen-ZhuLOW17}, including the famous RAS method of Sinkhorn~\cite{sinkhorn1964relationship} and improved analysis of it~\cite{KalantariK93,KalantariLRS08}. %

In this paper, we focus on a different special case of the symmetric matrix scaling problem where $A$ is what is known as an \emph{$M$-matrix}; a symmetric matrix is an $M$-matrix if it is \emph{positive definite}, i.e., $x^\top A x > 0$ for all nonzero $x \in \R^d$, and its off-diagonal entries are non-positive, i.e., $A_{ij} \leq 0$ for all $i \neq j$~\cite{BermanP94}.
Our goal is to shed light on the complexity of this problem and compute an $\eps$-approximate scaling (i.e., $x \in \R^d_{>0}$ such that $\|X A X \vec{1} - \vec{1}\|_{2} \leq \eps$) as efficiently as possible.
We measure efficiency in terms of both the sequential running time and the parallel depth, and focus on algorithms which solve the scaling problem to \emph{high accuracy} whose efficiency depends logarithmically on the accuracy parameter $\eps$.

This \emph{symmetric $M$-matrix scaling problem} is of particular interest due to recent progress of \cite{Vladu25}.
Excitingly, \cite{Vladu25} showed that this special case can be solved particularly efficiently using interior point methods (IPMs). %
More specifically, \cite{Vladu25} gave an $\tilde{O}(d^{1/3})$
iteration\footnote{Throughout the paper, we use $\tilde{O}(\cdot)$ to hide poly-logarithmic factors in the dimension $d$.} IPM that solved this problem to high accuracy.
Each iteration of the IPM involves solving a linear system in a Laplacian matrix which can be implemented in $\tilde{O}(\nnz(A))$ time and
$\tilde{O}(1)$ parallel depth. %

This result is especially notable as IPMs typically require a larger $\Omega(d^{1/2})$ iterations. There are faster algorithms which solve this problem in almost-linear time, e.g., through the convex-cost flow algorithm of \cite{ChenKLPGS22} (see \Cref{appendix:convex-flow}).
However, this algorithm is highly sequential and requires an almost-linear $\Omega(d)$ depth, as opposed to the $\tilde{O}(d^{1/3})$ depth of \cite{Vladu25}.

Consequently, the central question we consider in this paper is: \emph{are there even more efficient
algorithms for the symmetric $M$-matrix scaling problem, and can such algorithms be highly parallel?} 

\paragraph{Our Results.} In this paper we answer this question in the affirmative. We give an algorithm that not only improves on the almost-linear running time of \cite{ChenKLPGS22} but also runs in poly-logarithmic depth, improving upon \cite{Vladu25}.

\begin{theorem} 
  There is an algorithm that,
  given a symmetric $M$-matrix $A \in \R^{d \times d}$, computes an $x \in \R^d_{>0}$ such that $\|X A X \vec{1} - \vec{1}\|_2 \leq \eps$ with high probability in
  $\tilde{O}(\nnz(A) \cdot\log^3 \kappa(A) \log(\frac{1}{\eps})\log\log(\frac{1}{\eps}))$ time
  and
  $\tilde{O}(\log^3 \kappa(A) \log(\frac{1}{\eps})\log\log(\frac{1}{\eps}))$ parallel depth.
  \label{thm:main-randomized}
\end{theorem}

Excitingly, we obtain this result by essentially carefully reducing the problem to Laplacian system solving.
We first compute a scaling that makes $A$ \emph{(symmetric) diagonally dominant (SDD)}, i.e., $A_{ii} \geq \sum_{j \neq i}|A_{ij}|$ for all $i \in [d]$; this is a well-studied subroutine for $M$-matrices, and such a scaling is known to be computable with $\tilde{O}(1)$ Laplacian system solves.
From this scaling, we then build upon \cite{Vladu25} to obtain an $\tilde{O}(1)$-iteration IPM where again, the cost of each iteration is dominated by the cost of solving a Laplacian system. 

Altogether this shows that symmetric $M$-matrix scaling is no harder than Laplacian system solving in both sequential running time and parallel depth by more than an $\tilde{O}(1)$ factor.
Our main technical result is thus the following.

\begin{theorem}[Informal; see \Cref{thm:alg-matrix-scaling}]
  There is an algorithm that,
  given a symmetric $M$-matrix $A \in \R^{d \times d}$, computes an $x \in \R^d_{>0}$ such that $\|X A X \vec{1} - \vec{1}\|_2 \leq \eps$ through a sequence of $\tilde{O}(\log^2\kappa(A) + \log\log(1/\eps))$ SDD linear system solves, each of which involves $O(\nnz(A))$ non-zero entries and is solved to an $\Omega(\eps/\poly(d,\kappa(A)))$ accuracy.
\end{theorem}

\paragraph{Additional Related Work.}

Over the past decades there have been numerous advances on a variety of graph theoretic, linear algebraic, and convex optimization problems, many of which feature variants of $M$-matrices which we briefly survey below.

\paragraph*{\normalfont\itshape Laplacian and $M$-Matrix Solvers.}
$M$-matrices belong to a well-studied hierarchy of matrix classes that are of high algorithmic interest.
Indeed, $M$-matrices are precisely rescalings of positive-definite SDD $Z$-matrices, the latter of which encompasses Laplacian matrices.
Since the seminal work of Spielman and Teng~\cite{SpielmanT04}, there has been a long line of research on obtaining highly efficient and parallel solvers for Laplacian systems~\cite{KoutisMP10,KoutisMP11,LeeS13,CohenKMPPRX14,PengS14,KyngS16,JambulapatiS21,SachdevaZ23}.
Through a simple reduction, this implies efficient solvers for SDD systems as well.
Finally, \cite{AhmadinejadJSS19} obtained a reduction from making $M$-matrices SDD to SDD system solving itself.
This shows that, up to $\tilde{O}(1)$ factors, $M$-matrix systems, as well as a wide range of Perron--Frobenius type problems, are as easy as solving Laplacians, in both running time and parallel depth.

\paragraph*{\normalfont\itshape Graph Algorithms.}
Apart from the inherent linear algebraic applications, the general interest in the class of matrices mentioned above stems from its connection to fundamental graph problems through what is called the \emph{Laplacian paradigm}.
Indeed, over the past decades, efficient Laplacian solvers combined with various other advancements have led to many improved algorithms for a variety of graph problems, including maximum matching, minimum cost flow, spanning tree sampling, graph sparsification, and many more~\cite{DaitchS08,SpielmanS08,KelnerM09,ChristianoKMST11,KelnerMP12,LeeS14,KathuriaLS20,BrandLNPSS0W20,BrandLLSS0W21,BrandGJLLPS22,JiangKLS25}.

\paragraph*{\normalfont\itshape Interior Point Methods.}
IPMs are powerful frameworks for constrained optimization problems and have been extremely well studied.
Interestingly, they also played a key role in the aforementioned graph algorithm advances.
Classic self-concordance analysis suggests that if the constraint set admits a so-called $\nu$-self-concordant barrier, then there is an $\tilde{O}(\sqrt{\nu})$-iteration IPM; moreover, there is always an $O(d)$-self-concordant barrier in $\R^d$~\cite{NesterovN94}.
In a breakthrough result, \cite{LeeS14} showed that there exists an $\tilde{O}(\sqrt{r})$-iteration IPM where $r \leq d$ is the ``inherent dimension'' of the constraint set: for a linear program with $n$ variables and $d$ constraints, this reduces the number of iterations from $\tilde{O}(\sqrt{n})$ to $\tilde{O}(\sqrt{d})$.\footnote{Such a ``universal'' barrier was known to exist before \cite{LeeS14} but it was computationally intractable.}
Building on \cite{LeeS14}, there have been many works on obtaining improved algorithms for linear programming and related problems~\cite{BrandLSS20,BrandLLSS0W21,Jiang0WZ21}; however, none of them improved the $\tilde{O}(\sqrt{d})$ iteration bound.
There are results that established iteration lower bounds for certain types of IPMs under certain settings~\cite{ToddY96,todd2020lower}, and overall there have been very few IPM instances that obtained an iteration bound beyond self-concordance theory~\cite{CohenLMPS16,BubeckCLL18,Vladu25}.

\paragraph{Preliminaries.}

We use $\vec{0}$ and $\vec{1}$ to denote the all-zero and all-one vectors whose dimension shall be clear from context.
For vector $x \in \R^d$ we let $X \defeq \diag(x) \in \R^{d \times d}$ be the diagonal matrix associated with $x$.
We write $x^2 \in \R^d$, $1/x \in \R^d$, $x/y\in\R^d$ for $y\in \R^d$ as applying squaring, inverting, and division to $x$ entry-wise.
We also let comparison between vectors be done entry-wise, e.g., $\vec{0} < x \leq y$, in which case the vectors shall be of the same dimension.

Matrices are assumed to be symmetric throughout the paper.
Let $\nnz(A)$ denote the number of nonzero entries in a matrix $A$.
We assume without loss of generality that $\nnz(A) = \Omega(d)$.
We generally work with a \emph{$Z$-matrix} for which the off-diagonal entries are non-positive.
As mentioned previously, an \emph{$M$-matrix} is a positive definite $Z$-matrix.

For a positive definite matrix $A$, we let $\kappa(A) \defeq \lambda_{\max}(A) / \lambda_{\min}(A)$ denote its \emph{condition number}.
For $x \in \R^d_{>0}$ we have $\kappa(X) = \max_{i \in [d]}x_i / \min_{i \in [d]}x_i$.
We let $\norm{A}_p \defeq \max_{\norm{x}_p = 1}\norm{Ax}_p$ denote the $\ell_p$-operator norm of $A$.
In particular, $\norm{A}_{\infty} = \max_{i \in [d]}\sum_{j \in [d]}|A_{ij}|$.
We work with different $\ell_p$-norms of $A$ almost interchangeably since they differ by at most a $\poly(d)$ factor.
The $\ell_p$-norms of $A$ and $A^{-1}$ are also known to approximate the top/bottom eigenvalues of $A$.
In particular, our algorithm needs access to some $\ell_p$-norm or top eigenvalue of $A$ which are hard to compute exactly, in which case we may use the easy-to-compute proxy $\norm{A}_{\infty}$.

\begin{fact}
  For a positive definite matrix $A$ we have $\norm{A}_2 = \lambda_{\max}(A)$, $\norm{A^{-1}}_2 = 1/\lambda_{\min}(A)$, and $\norm{A}_q/d \leq \norm{A}_p \leq \norm{A}_q \cdot d$ for any $1 \leq p, q \leq \infty$.
\end{fact}

Our analysis heavily relies on the following lemma that follows from the nonnegativity of $M$-matrix inverses.

\begin{lemma}
  For an M-matrix $A$ with $A\vec{1} \geq c \cdot \vec{1}$ where $c > 0$, we have $\|A^{-1}\|_{\infty} \leq 1/c$.
  \label{lem:mmatrix-norm}
\end{lemma}

\begin{proof}
    As
    $A^{-1}$ is entry-wise nonnegative~\cite{BermanP94},  multiplying both sides by $A^{-1}$, we have $\vec{1} \geq c A^{-1} \vec{1}$.
    The $\ell_{\infty}$-norm of a nonnegative matrix is equal to the maximum row sum, so $\|A^{-1}\|_{\infty} = \max_{i \in [d]} [A^{-1} \vec{1}]_i \leq 1/c$.
\end{proof}

Let $H \in \R^{d \times d}$ be positive definite.
We often consider solving linear systems in $H$ approximately.
For $v \in \R^d$ and $\eta \geq 0$, let $\A_{\eta}(H, v) \defeq \{u \in \R^d : \norm{u - H^{-1}v}_2 \leq \eta \norm{v}_2 \norm{H^{-1}}_2\}$ consist of \emph{$\eta$-approximate solutions to $H^{-1}v$}.\footnote{This is the natural scale-invariant notion of approximation for solving linear systems.}

\paragraph{AI Usage Disclosure.}
The main ideas of the paper were developed primarily by the authors in tandem with general research discussions with LLMs regarding the problem.
LLMs were also used to assist with certain mathematical derivations, writing, formatting, and proofreading.
The authors assume full responsibility for all content.

%% file: overview.tex
\section{Technical Overview} \label{sec:overview}

In this section we motivate our algorithm and introduce the key technical insights underlying it. 
Our approach is inspired by and heavily builds on the IPM  of \cite{Vladu25}.
We first describe how this IPM is set up and analyzed in \cite{Vladu25}, and then we give intuition on how we obtain a faster algorithm.
Note that the focus of \cite{Vladu25} is on obtaining a $d^{1/3}$-iteration algorithm; however, we will only discuss and rely on a (rather standard) $d^{1/2}$-iteration analysis.

\paragraph{A Path-Following Approach to Matrix Scaling.}
To solve the matrix scaling problem, \cite{Vladu25} considers the function $f_{\mu}(x) \defeq \frac{1}{\mu} \cdot (\frac{1}{2}x^\top A x - b^\top x) - \sum_{i \in [d]}\ln x_i$ for some vector $b \in \R^d$. \cite{Vladu25} shows that to solve the matrix scaling problem, it suffices to compute an (approximate) minimizer of $f_{\mu}$ for a sufficiently large $\mu$.
To see why this is the case, let $x_\mu  \in \R^d_{> 0}$ denote the optimizer of $f_{\mu}$. The optimality conditions for $x_\mu$ imply that $\vec{0} = \nabla f_{\mu}(x_\mu) = \frac{1}{\mu} (A x_\mu - b) - 1 / x_\mu$,
and therefore $\tilde{X}_\mu A \tilde{X}_\mu \vec{1} = \vec{1} + \frac{1}{\sqrt{\mu}}\tilde{X}_\mu b$, where $\tilde{x}_\mu \defeq x_\mu / \sqrt{\mu}$.
It can then be shown that when $\mu$ is sufficiently large, the contribution of $ \tilde{X}_\mu b$ is sufficiently small and $\tilde{X}_\mu A \tilde{X}_\mu \vec{1}$ is close to $\vec{1}$.

To compute an approximate $x_\mu$ for large $\mu$,  \cite{Vladu25} considered the central path, i.e., the curve $\{x_{\mu} : \mu > 0\}$, and applied what is called a \emph{path-following algorithm} to approximately follow this curve for increasing values of $\mu$. Starting from an initial small $\mu_0$ for which the minimizer $x_{\mu_0}$ of $f_{\mu_0}$ can be easily computed, the algorithm iteratively increases $\mu$ while staying on the central path.
This is often done by arguing that as long as $x_{\mu}$ is sufficiently ``close'' to $x_{\mu'}$ where $\mu' > \mu$, we can minimize $f_{\mu'}$ efficiently by taking a few Newton steps from $x_{\mu}$.
With this, the choice of the distance measure is thus crucial to the analysis of the IPM:
If we could increase the region that is considered ``close'' to the central path, then we can increase $\mu$ more aggressively and decrease the overall number of iterations.

\paragraph{The Standard $O(d^{1/2})$ Analysis.}
Since we want to minimize $f_{\mu}$ by Newton steps, to measure the distance from some $x$ to the central path, it is natural to consider the size of the Newton step at $x$.
This is captured by the notion of the \emph{congestion vector}
\[ \rho_{x,\mu} \defeq X^{-1} \nabla^2 f_{\mu}(x)^{-1} \nabla f_{\mu}(x) = \frac{1}{\mu} (I + A_{\mu}(x))^{-1} X (A x - b - \mu/x), \]
where $A_{\mu}(x) \defeq \mu^{-1} X A X$.
As is standard in IPM literature, both \cite{Vladu25} and our algorithm
use the norm of $\rho_{x,\mu}$ as the measure of how close $x$ is to the central path.
Let $x^\prime \defeq x - \nabla^2 f_{\mu}(x)^{-1} \nabla f_{\mu}(x) = X (\vec{1} - \rho_{x,\mu})$ be the result of taking a Newton step from $x$.
It can be shown
that
\[
  \nabla f_{\mu}(x^\prime) = -\frac{\rho^{2}_{x,\mu}}{x^\prime}
\]
and therefore
\begin{equation}
  \rho_{x^\prime, \mu} = -(I + A_{\mu}(x^\prime))^{-1} \rho_{x,\mu}^2.
  \label{eq:new-congestion}
\end{equation}
Since $A_{\mu}(x^\prime)$ is PSD, we have that $\|\rho_{x^\prime,\mu}\|_2 \leq \|\rho^2_{x,\mu}\|_2 \leq \|\rho_{x,\mu}\|_2^2$.
Consequently, if we use the $\ell_2$-norm of the congestion vector as the distance measure to the central path, as long as we stay within the region where $\|\rho_{x,\mu}\|_2 < 1/2$, we can quickly compute $x_{\mu}$ from $x$ using Newton steps.

We can analyze the algorithm using the $\ell_2$-norm. %
Consider a perfectly centered $x_{\mu}$.
We can increase $\mu$ to $\mu^\prime$ while still allowing efficient centering as long as $\|\rho_{x_{\mu}, \mu^\prime}\|_2 < 1/2$.
Suppose that $\mu^\prime = \mu / (1 - \delta)$.
We have that
\begin{equation}
\rho_{x_{\mu},\mu^\prime} = -\delta (I + A_{\mu^\prime}(x_{\mu}))^{-1} \vec{1}.
\label{eq:mu-change}
\end{equation}
Again, since $A_{\mu^\prime}(x_{\mu})$ is PSD, we have $\|\rho_{x_{\mu}, \mu^\prime}\|_2 \leq  \delta \sqrt{d}$.
In other words, we can increase $\mu$ by a $1 + \Theta(1/\sqrt{d})$ factor while maintaining the ability to move back to the central path.
This gives a $\tilde{O}(\sqrt{d})$-iteration method.

Notice that since $\|\rho_{x,\mu}^2\|_2 = \|\rho_{x,\mu}\|_4^2$, we actually only need a small $\ell_4$-norm of $\rho_{x,\mu}$ for centering (as opposed to the larger $\ell_2$-norm).
However, in general there is no better  bound on the $\ell_4$-norm of $(I + A_{\mu^\prime}(x_{\mu}))^{-1} \vec{1}$ other than that $\|(I + A_{\mu^\prime}(x_{\mu}))^{-1} \vec{1}\|_4 \leq \|(I + A_{\mu^\prime}(x_{\mu}))^{-1} \vec{1}\|_2$ in \eqref{eq:mu-change} if we only use that $A_{\mu^\prime}(x_{\mu})$ is PSD.
This $\ell_4$-versus-$\ell_2$ discrepancy is a common phenomenon in IPMs, which \cite{Vladu25} circumvented using a ``global'' amortized analysis involving $\ell_3$-norms (and that $A_{\mu^\prime}(x_{\mu})$ is an $M$-matrix).\footnote{There have been other attempts to address the gap between $\ell_4$- and $\ell_2$-norm in other contexts. For example, for maximum flow, \cite{ChristianoKMST11,Madry13,Madry16} argued that when $\norm{\rho}_4 \approx \norm{\rho}_2$ then one can perturb the graph and increase a potential which implies that the number of such iterations is bounded.
On the other hand, when $\norm{\rho}_4 \ll \norm{\rho}_2$, we can take a much larger IPM step.

Both these algorithms and \cite{Vladu25} are more ``global'' compared to the standard $d^{1/2}$ analysis in the sense that they do not obtain an improved uniform bound on the progress of each ``local'' step; instead, they rely on a ``global'' argument that involves all steps in the IPM algorithm.
}

\paragraph{Our Approach.}

In this paper, we go beyond the $\ell_3$-norm and analyze the IPM using the $\ell_{\infty}$-norm of the congestion vector as a measure of centrality. %
The goal is to show that $\|(I + A_{\mu}(x))^{-1}\|_{\infty} \leq O(1)$ using the stronger fact that $A_{\mu}(x)$ is an $M$-matrix;
indeed, we see that such a bound combined with \eqref{eq:new-congestion} and \eqref{eq:mu-change} would suggest that $\norm{\rho_{x_{\mu},\mu'}}_{\infty} = O(\delta)$ and therefore setting $\delta = \Theta(1)$ still allows efficient centering.

To that end, we leverage an important property that the $\ell_{\infty}$-norm of $M^{-1}$ for an $M$-matrix $M$ is simply its maximum row/column sum, i.e., $\max_{i \in [d]}[M^{-1} \vec{1}]_i$.
Let us take a closer look at $(I + A_{\mu}(x))$ which is also an $M$-matrix.
We have
\[
  (I + A_{\mu}(x))\vec{1} = \vec{1} + \frac{1}{\mu} X A x = \vec{1} + (I + A_{\mu}(x)) \rho_{x,\mu} + \frac{1}{\mu}Xb + \vec{1}.
\]
Multiplying both sides by $(I + A_{\mu}(x))^{-1}$, we get
\begin{equation}
  \vec{1} = (I + A_{\mu}(x))^{-1} \left(\vec{2} + \frac{1}{\mu}Xb\right)  + \rho_{x,\mu} \implies (I + A_{\mu}(x))^{-1}\vec{1} = \frac{1}{2} \cdot \left(\vec{1} - \rho_{x,\mu} - \frac{1}{\mu}(I + A_{\mu}(x))^{-1} Xb\right).
  \label{eq:inf-norm}
\end{equation}

\paragraph{The Key Observation.}
Given the above calculation, notice that if the vector $b$ happens to be entry-wise nonnegative, then the right-hand side of \eqref{eq:inf-norm} is entry-wise upper-bounded by $\frac{1}{2}(\vec{1} - \rho_{x,\mu})$.
We can therefore bound $\|(I + A_{\mu}(x))^{-1}\|_{\infty} \leq \frac{1 + \|\rho_{x,\mu}\|_{\infty}}{2}$.
When $\|\rho_{x,\mu}\|_{\infty} < 1/2$, we get $\|(I + A_{\mu}(x))^{-1}\|_{\infty} \leq 1$ which would lead us to an $\tilde{O}(1)$-iteration algorithm.

Obtaining our fast convergence rate thus boils down to the vector $b$.
Our key observation is that for the algorithm to correctly find the scaling, \emph{the choice of $b$ can almost be arbitrary.}
Indeed, the only place where $b$ shows up in the analysis is the initialization and termination of the IPM.
Most notably, the vector $b$ should be chosen so that we can easily compute an initial point $x_0$ that is centered for some $\mu_0$.
In \cite{Vladu25}, the vector $b$ is initialized to $A\vec{1} - \vec{1}$ so that $x_0 \defeq \vec{1}$ is centered for $\mu_0 \defeq 1$;
however, this $b$ is not necessarily a nonnegative vector.

We observe that there is an alternative way of initializing the path-following algorithm that will make $b$ nonnegative.
We leverage the fact that the $M$-matrix $A$ admits a vector $y \in \R^d_{>0}$ such that $Y A Y \vec{1} \geq \gamma \cdot \vec{1} > \vec{0}$ for some $\gamma$ bounded from below.
Such a vector $y$ is related to the well-studied problems in Perron--Frobenius theory and solving linear systems in $M$-matrices, and \cite{AhmadinejadJSS19} gave an algorithm that computes $y$ by solving $\tilde{O}(1)$ Laplacian linear systems.
With $y$, we can initialize our path-following algorithm with $x_0 \defeq y / \sqrt{\gamma}$ (which satisfies $X_0 A X_0 \vec{1} \geq \vec{1}$) and $\mu_0 \defeq 1$, which is centered if we set $b \defeq X_0^{-1} (X_0 A x_0 - \vec{1})$ which is now entry-wise nonnegative.
This completes the analysis of our $\tilde{O}(1)$ iteration algorithm.

\paragraph{Our Actual Analysis.}

We presented an analysis above that closely follows the notation and setup in \cite{Vladu25} for building intuition and illustrating the main technical insights.
However, in \Cref{sec:scaling}, we
use a slightly different function
\begin{equation}
f_{\mu}(x) \defeq -\mu \cdot b^\top x + \left(\frac{1}{2}x^\top A x - \sum_{i \in [d]}\ln x_i\right).
\label{eq:new-f}
\end{equation}
Ignoring the $-\mu \cdot b^\top x$ term (which happens when we decrease $\mu$ to a small number), the optimality condition of \eqref{eq:new-f} directly gives $A x = 1/x \implies X A X \vec{1} = \vec{1}$.
We also use a slightly different centrality measure $\norm{X\nabla f_{\mu}(x)}_{\infty}$ instead of $\norm{\rho_{x,\mu}}_{\infty}$ (see more discussion in \Cref{sec:scaling}).
We think this new version is mathematically cleaner to work with,
but the basis of our improvement remains the same as aforementioned.

\paragraph{Comparison with Quadratic Optimization and Limitations of Our Approach.}

\cite{Vladu25} also considered the closely related quadratic optimization problem which for a given matrix $A$ and vector $b$ asks for $\min_{x \geq 0}\frac{1}{2}x^\top A x - b^\top x$.
When $A$ is a symmetric $M$-matrix, \cite{Vladu25} gave a similar $\tilde{O}(d^{1/3})$-iteration IPM for this problem, which uses the solution to the matrix scaling problem as the initial point.

While our algorithm improves the $M$-matrix scaling problem,
our key ingredient fails for
the
quadratic optimization problems because we can no longer choose a nonnegative vector $b$ freely.
In general, however, there does not seem to be a direct reduction from one to the other, and the two problems seem fairly incomparable: Indeed, one has a linear term $b^\top x$ on top of $\frac{1}{2}x^\top A x$ and the other has $\sum_i \ln(x_i)$.
We leave obtaining a deeper understanding of the two problems as an interesting open problem.\footnote{Interestingly, quadratic optimization for $M$-matrices can be reduced to the \emph{$\ell_2$-norm flow diffusion} problem for which \cite{ChenPW21} gave a \emph{near-linear} time algorithm. In contrast, a similar ``graphical'' approach reduces matrix scaling to convex-cost flow (as in \Cref{appendix:convex-flow}) which seems to be a much harder problem and for which only almost-linear time algorithms are known~\cite{ChenKLPGS22}. As discussed earlier, the flow approach was also the best known sequential algorithm for $M$-matrix scaling prior to our work.}

%% file: scaling.tex
\section{Efficient Parallel $M$-Matrix Scaling}\label{sec:scaling}

In this section we provide a formal analysis of our algorithm.
In \Cref{sec:path-following} we show that given an appropriate initial $x_0$, our path-following algorithm computes the desired scaling by solving $\tilde{O}(1)$ linear systems over diagonally dominant matrices.
In \Cref{sub:predcondition} we show how to compute the initial $x_0$ which then allows us to prove \Cref{thm:main-randomized} using the fact that diagonally dominant linear systems can be solved in near-linear time and poly-logarithmic parallel depth.

\subsection{\texorpdfstring{$\ell_\infty$}{L_infinity}-Path Following Algorithm}\label{sec:path-following}

Here we describe our $\ell_{\infty}$-path following algorithm.
For a fixed nonnegative vector $b\in \R^d_{\geq 0}$ and all $\mu \geq 0$, let $f_\mu(x)\defeq -\mu b^\top x+\frac{1}{2}x^\top A x-\sum_{i\in[d]}\log x_i$, and let $A(x)\defeq XAX$.\footnote{Technically, the definitions of $f$ and $\Delta_{\mu}^{\infty}(x)$ depend on the choice of $b$ which is not made explicit until the end of the analysis. However, for notational convenience, we may simply refer to $f$ and $\Delta_{\mu}^{\infty}(x)$ as parameterized by the vector $b$ given to each lemma statement.} 
Thus, we have 
\[\nabla f_\mu(x)=Ax-\mu b-X^{-1}\vec{1}\quad\mbox{and}\quad \nabla^2 f_\mu(x)=A+X^{-2}=X^{-1}(I+A(x))X^{-1}.\]
We define the $\ell_\infty$-measure of centrality
\[
\Delta^\infty_\mu(x)\defeq \norm{X\nabla f_\mu(x)}_\infty
\text{ for all }
x \in \R^{d}_{>0}\,,
\]
which satisfies $\Delta^\infty_\mu(x)=\norm{X(Ax-\mu b-X^{-1}\vec{1})}_\infty$.
Note that $X\nabla f_{\mu}(x) = (I + A(x)) \rho_{x,\mu}$ where $\rho_{x,\mu}$ is the congestion vector from \Cref{sec:overview} (and in \eqref{eq:congestion} later).

\Cref{alg:path-following} presents the pseudocode of our algorithm.
The algorithm is given an initial point $x_0 > \vec{0}$ which is assumed to satisfy $X_0AX_0\vec{1}\geq\vec{1}$.
We set $b=X_0^{-1}(X_0AX_0\vec{1}-\vec{1})$, $\mu= 1$, and start with $x = x_0$ which by \Cref{lem:initialization} is perfectly centered with $\Delta^{\infty}_{\mu}(x) = 0$.
The algorithm has two phases.
In the first one, we repeatedly decrease $\mu$ and perform one approximate Newton step in $f_{\mu}$ to re-center until
$\mu$ is sufficiently small.
Once $\mu$ is small enough, we enter the second phase of the algorithm where we set $\mu$ all the way to $0$.
This only increases the centrality of $x$ slightly (see \Cref{lem:mu-set-zero}), and allows us to now directly decrease our objective $\Delta_{0}^{\infty}(x) = \norm{XAX \vec{1} - \vec{1}}_{\infty}$.
This is done by taking multiple Newton steps in $f_0$ where each of them decreases $\|XAX \vec{1} - \vec{1}\|_{\infty}$ quadratically.

\input{\pathfollowingalgorithm}

We first verify that our initialization is valid.

\begin{lemma}\label{lem:initialization}
	Suppose $A\in \R^{d\times d}$  and $x_0\in \R^d_{>0}$ satisfy $X_0AX_0\vec{1}\geq\vec{1}$.
	Let $\mu_0\defeq 1$ and $b\defeq X_0^{-1}(X_0AX_0\vec{1}-\vec{1})$.
	Then, $b \geq \vec{0}$ and $\Delta_{\mu_0}^{\infty}(x_0) = 0$.
\end{lemma}
\begin{proof}
    The fact that $b \geq \vec{0}$ follows immediately from its definition and that $X_0 A X_0 \vec{1} \geq \vec{1}$.
    We can also compute
	\[
    \begin{split}
      \Delta_{\mu_0}^{\infty}(x_0)
      &= \norm{X_0\nabla f_{\mu_0}(x_0)}_{\infty}
      = \norm{X_0AX_0\vec{1}-\mu_0X_0b-\vec{1}}_{\infty} \\
      &= \norm{X_0AX_0\vec{1}-(X_0AX_0\vec{1}-\vec{1})-\vec{1}}_{\infty}
      = 0. \qedhere
    \end{split}
	\]
\end{proof}

We now analyze the main loops of the algorithm.
We start by connecting the centrality $\Delta_{\mu}^{\infty}(x)$ to the $\ell_{\infty}$-norm of $(I + A(x))^{-1}$.

\begin{lemma}\label{lem:inverse-bound-from-delta}
	Suppose $b \in \R^d_{\geq 0}$ and $A$ is a symmetric $M$-matrix. For all $x \in \R^d_{> 0}$ such that $\Delta^\infty_{\mu}(x) < 2$ for some $\mu \geq 0$ we have
	\[
	\norm{(I + A(x))^{-1}}_\infty \leq \frac{1}{2 - \Delta^\infty_{\mu}(x)}\,.
	\]
\end{lemma}
\begin{proof}
	Note that
	\[
	X\nabla f_\mu(x)=A(x)\vec{1}-\mu Xb-\vec{1}.
	\]
	Since $X\nabla f_\mu(x)\geq -\Delta^\infty_\mu(x)\vec{1}$ and $b\geq \vec{0}$, we have
    \begin{equation}
	A(x)\vec{1}
	\geq
	\mu Xb+\vec{1}-\Delta^\infty_\mu(x)\vec{1}
	\geq
	(1-\Delta^\infty_\mu(x))\vec{1}.
    \label{eq:matrix-bound}
    \end{equation}
	Consequently, $(I+A(x))\vec{1}\geq (2-\Delta^\infty_\mu(x))\vec{1}$, and the result follows from \Cref{lem:mmatrix-norm}.
\end{proof}

To analyze the Newton steps, as in \cite{Vladu25}, we define the \emph{congestion vector}
\begin{equation}
\rho_{x,\mu} \defeq X^{-1} \nabla^2 f_{\mu}(x)^{-1} \nabla f_{\mu}(x).
\label{eq:congestion}
\end{equation}
The following claim expresses the gradient after one Newton step in terms of the congestion vector.
This is exactly the same as the calculation that was done in \cite[Equation (3)]{Vladu25}, but since we use a different $f_{\mu}$ function, we redo it here for completeness.

\begin{claim}[{\cite[Equation (3)]{Vladu25}}]\label{claim:congestion-vector}
Let $x \in \R^d_{> 0}$ and $\mu\geq 0$.
We have, for $y \defeq x - h$ where $h \defeq (\nabla^2 f_\mu(x))^{-1} \nabla f_\mu(x)$,
that
if $y>\vec{0}$, then
 $\nabla f_{\mu}(y) = -\rho^{2}_{x,\mu}/y$.
\end{claim}
\begin{proof}
    By direct calculation from the definition of $y$ and $f_{\mu}$, we have
	\begin{align*}
		\nabla f_\mu(y)
		&=
		\nabla f_\mu(x)+\nabla f_\mu(y)-\nabla f_\mu(x)\\
		&=
		\nabla^2 f_\mu(x)X\rho_{x,\mu}
		+\nabla f_\mu(y)-\nabla f_\mu(x)\\
		&=
		\nabla^2 f_\mu(x)X\rho_{x,\mu}
		+\nabla^2 f_\mu(x)(y-x)
		+\left(\nabla f_\mu(y)-\nabla f_\mu(x)-\nabla^2 f_\mu(x)(y-x)\right)\\
		&=
		\nabla f_\mu(y)-\nabla f_\mu(x)-\nabla^2 f_\mu(x)(y-x)\\
		&=
		-\frac{1}{y}+\frac{1}{x}-\frac{y-x}{x^2}
		=
		-\frac{(y-x)^2}{x^2y}
		=
		-\frac{\rho_{x,\mu}^2}{y}. \qedhere
	\end{align*}
\end{proof}

We now use \Cref{lem:inverse-bound-from-delta,claim:congestion-vector} to analyze the centrality after decreasing $\mu$ and taking one Newton step.

\begin{lemma}\label{lem:mu-update}
	Suppose $b\in\R_{\geq0}^d$ and $A$ is a symmetric $M$-matrix.
	Let $x\in\R_{>0}^d$ and $\mu\geq 0$ such that  $\Delta_\mu^\infty(x)\leq 2/3$.
    Then, for any $\delta \in [0, 1)$, we have, for $y \defeq x - h$ where
	$h\defeq(\nabla^2f_{(1-\delta)\mu}(x))^{-1}\nabla f_{(1-\delta)\mu}(x)$, that
	$y > \vec{0}$ as well as
	\[
    \norm{X^{-1}h}_{\infty} \leq \left(\delta+(1-\delta)\frac{\Delta_\mu^\infty(x)}{2-\Delta_\mu^\infty(x)}\right)
    \quad\text{and}\quad
	\Delta_{(1-\delta)\mu}^\infty(y)\leq\left(\delta+(1-\delta)\frac{\Delta_\mu^\infty(x)}{2-\Delta_\mu^\infty(x)}\right)^2.
	\]
\end{lemma}
\begin{proof}
	By \Cref{lem:inverse-bound-from-delta},
	\[
	\norm{(I+A(x))^{-1}}_\infty\leq\frac{1}{2-\Delta_\mu^\infty(x)}\leq\frac{3}{4}.
	\]
	Moreover, $(I+A(x))^{-1}\geq0$, and hence
	\[
	\vec{0}\leq (I+A(x))^{-1}\vec{1}\leq\frac{3}{4}\vec{1}
    \implies
	\norm{\left(I-2(I+A(x))^{-1}\right)\vec{1}}_\infty\leq1.
	\]
	Since
	\[
	X\nabla f_{(1-\delta)\mu}(x)
	=
	(1-\delta)X\nabla f_\mu(x)+\delta(A(x)-I)\vec{1}
	\]
	and
	\[
	(I+A(x))^{-1}(A(x)-I)
	=
	(I+A(x))^{-1}\left((I+A(x))-2I\right)
	=
	I-2(I+A(x))^{-1},
	\]
	we have
	\[
	\begin{split}
		X^{-1}h
		&=(I+A(x))^{-1}X\nabla f_{(1-\delta)\mu}(x)\\
		&=(1-\delta)(I+A(x))^{-1}X\nabla f_\mu(x)
		+\delta\left(I-2(I+A(x))^{-1}\right)\vec{1}.
	\end{split}
	\]
	Therefore,
	\[
	\begin{split}
		\norm{X^{-1}h}_\infty
		&\leq
		(1-\delta)\norm{(I+A(x))^{-1}X\nabla f_\mu(x)}_\infty
		+\delta\norm{\left(I-2(I+A(x))^{-1}\right)\vec{1}}_\infty \\
		&\leq
		(1-\delta)\norm{(I+A(x))^{-1}}_\infty\norm{X\nabla f_\mu(x)}_\infty+\delta
		=
		(1-\delta)\norm{(I+A(x))^{-1}}_\infty\Delta_\mu^\infty(x)+\delta \\
		&\leq
		(1-\delta)\frac{\Delta_\mu^\infty(x)}{2-\Delta_\mu^\infty(x)}+\delta<1,
	\end{split}
	\]
	where the second step uses $\norm{(I-2(I+A(x))^{-1})\vec{1}}_\infty\leq1$, the fourth step uses \Cref{lem:inverse-bound-from-delta}, and the last step uses $\Delta_\mu^\infty(x)\leq 2/3$ and $\delta<1$.
	This shows that $y\defeq x-h$ satisfies $y>\vec{0}$.

    To bound $\Delta^{\infty}_{(1-\delta)\mu}(y)$,
    note that
	since $\rho_{x,(1-\delta)\mu}=X^{-1}h$, by \Cref{claim:congestion-vector}, we have $\nabla f_{(1-\delta)\mu}(y)=-\rho_{x,(1-\delta)\mu}^2/y$.
	Therefore,
	\[
	\begin{split}
		\Delta_{(1-\delta)\mu}^\infty(y)
		&=
		\norm{Y\nabla f_{(1-\delta)\mu}(y)}_\infty
		=
		\norm{\rho_{x,(1-\delta)\mu}^2}_\infty
		=
		\norm{\rho_{x,(1-\delta)\mu}}_\infty^2=
		\norm{X^{-1}h}_\infty^2\\
		&
		\leq
		\left(\delta+(1-\delta)\frac{\Delta_\mu^\infty(x)}{2-\Delta_\mu^\infty(x)}\right)^2.
        \qedhere
	\end{split}
	\]
\end{proof}

\Cref{lem:mu-update} assumes that we compute the step $h$ exactly.
However, in the actual algorithm we will only use an approximate linear system solver, so the step $h$ will only be approximately $(\nabla^2 f_{(1-\delta)\mu}(x))^{-1} \nabla f_{(1-\delta)\mu}(x)$ which we need to account for.
We defer the proof of the following approximate version of \Cref{lem:mu-update} to \Cref{appendix:omitted_proof}.

\begin{restatable}{lemma}{approxstepunified}\crefalias{theorem}{lemma}\label{lem:approx-step-unified}
	Suppose $A$ is a symmetric $M$-matrix.
    Let $x_0 \in \R^{d}_{>0}$ be such that $X_0AX_0\vec{1} \geq \vec{1}$ and let $b\defeq X_0^{-1}(X_0AX_0\vec{1} - \vec{1})$.
Let $x\in\R^d_{>0}$ and $\mu \in [0, 1]$ such that
$\Delta_\mu^\infty(x)\leq 2/3$.
Then, for any $\delta \in [0, 1/2]$ and $h' \in \A_{\eta}(\nabla^2 f_{(1-\delta)\mu}(x), \nabla f_{(1-\delta)\mu}(x))$ with
\[
  \eta \defeq \frac{\Delta_{\mu}^{\infty}(x) + \delta}{10^4 d q^3}\quad\text{where}\quad q \defeq \norm{A}_{\infty}\norm{x_0}_2^2,
\]
letting $y' \defeq x - h'$ we have $y' > \vec{0}$ and
\[
  \Delta_{(1-\delta)\mu}^{\infty}(y') \leq \frac{9}{8}\left(\delta + (1-\delta)\frac{\Delta_{\mu}^{\infty}(x)}{2 - \Delta_{\mu}^{\infty}(x)}\right)^2.
\]

\end{restatable}

After the first while loop in \Cref{alg:path-following}, we decrease $\mu$ drastically to $0$ and begin the second phase.
We show that as long as $\mu$ is sufficiently small and $x$ is sufficiently centered for $f_{\mu}$, $x$ remains centered for $f_0$ as well.
This provides a good upper bound on $\Delta^{\infty}_0(x) = \|XAX\vec{1} - \vec{1}\|_{\infty}$ which the second phase of our algorithm then minimizes.

\begin{lemma}\label{lem:mu-set-zero}
For any $x \in \R^{d}_{>0}$, if
$\Delta_\mu^\infty(x)\leq1/2$ for some
$\mu \leq \frac{\sqrt{\lambda_{\min}(A)}}{12\sqrt{d}\norm{b}_2}$,
then $\norm{XAX\vec{1}-\vec{1}}_\infty\leq 2/3$.
\end{lemma}
\begin{proof}
	We have
    \begin{equation}
	\norm{XAX\vec{1}-\vec{1}}_\infty
	\leq
	\norm{X(Ax-\mu b-X^{-1}\vec{1})}_\infty
	+
	\mu\norm{Xb}_\infty
	=
	\Delta_\mu^\infty(x)+\mu\norm{Xb}_\infty.
    \label{eq:identity}
    \end{equation}
	By definition of $\Delta_{\mu}^{\infty}(x)$, we have
	\[ 
	|x^\top Ax - (d + \mu b^\top x)| \leq d \Delta_{\mu}^{\infty}(x).
	\]
	Thus,
	\[
	\norm{x}_2^2 \lambda_{\min}(A) \leq x^\top Ax	\leq d(1+\Delta_{\mu}^{\infty}(x))+\mu\norm{b}_2\norm{x}_2.
	\]
	Since $\Delta_\mu^\infty(x)\leq 1/2$ and $\mu\norm{b}_2\leq \sqrt{\lambda_{\min}(A)}/(12\sqrt d)$, we have
	\[
	\norm{x}_2 \leq \sqrt{\frac{4d}{\lambda_{\min}(A)}}
    \implies
	\mu\norm{Xb}_\infty 
	\leq \mu\norm{x}_2\norm{b}_2 
	\leq \frac{\sqrt{\lambda_{\min}(A)}}{12\sqrt{d}\norm{b}_2} \cdot \sqrt{\frac{4d}{\lambda_{\min}(A)}}\norm{b}_2
	=\frac{1}{6}.
    \]
    Plugging this back to \eqref{eq:identity} proves the lemma.
\end{proof}

We also need to argue that we can compute the steps $h'$ in \Cref{alg:path-following} efficiently.
We show that they can be reduced to solving SDD linear systems with a moderate increase in the required accuracy.
To that end, let $\T_{\SDD}(m, \eta)$ and $\D_{\SDD}(m, \eta)$ be the running time and parallel depth of an algorithm that, given a positive definite SDD matrix $B$ with $\nnz(B) = m$ and $v \in \R^d$, computes a vector $u \in \A_{\eta}(B, v)$.
We prove the following lemma in \Cref{appendix:omitted_proof}.

\begin{restatable}{lemma}{solvercostpercall}\crefalias{theorem}{lemma}\label{lem:solver-cost-per-call}
	Suppose $b\in\R_{\geq0}^d$ and $A$ is a symmetric $M$-matrix.
	Let $x\in\R_{>0}^d$ and $\mu\geq 0$ such that $\Delta^{\infty}_{\mu}(x) \leq 2/3$.
    Then, for any $\mu' \geq 0$ and $\eta \geq 0$, we can compute an $h' \in \A_{\eta}(\nabla^2 f_{\mu'}(x), \nabla f_{\mu'}(x))$ in
    \[
      O\left(\T_{\SDD}\left(\nnz(A), \frac{\eta}{\kappa(X)^2}\right)\right)\;\text{time}\quad\text{and}\quad
      O\left(\D_{\SDD}\left(\nnz(A), \frac{\eta}{\kappa(X)^2}\right)\right)
      \;\text{parallel depth.}
    \]
\end{restatable}

We will also rely on the following claim regarding the iterates of \Cref{alg:path-following} whose proof we also defer to \Cref{appendix:omitted_proof}.

\begin{restatable}{claim}{Xentry}\crefalias{theorem}{claim}\label{claim:X-entry}
  Suppose $A$ is a symmetric $M$-matrix.
  Let $x_0 \in \R^d_{>0}$ be such that $X_0AX_0 \vec{1} \geq \vec{1}$ and let $b \defeq X_0^{-1}(X_0AX_0 \vec{1} - \vec{1})$.
  For any $x \in \R^d_{>0}$ such that $\Delta^{\infty}_{\mu}(x) \leq 2/3$ for some $\mu \in [0, 1]$, we have
  \[ XAx \geq \frac{1}{3}\vec{1},\quad x \leq \frac{5}{3}x_0,\quad\text{and}\quad\kappa(X) \leq 10q,\quad\text{where}\quad q \defeq \norm{A}_{\infty}\norm{x_0}_2^2.\]
\end{restatable}

Now, we are ready to analyze \Cref{alg:path-following}.

\begin{theorem}\label{thm:path-following}
	Given a symmetric $M$-matrix $A\in \R^{d\times d}$, a vector $x_0\in \R^d_{>0}$ such that $X_0AX_0 \vec{1} \geq \vec{1}$, and a precision parameter $\eps\in(0,1)$, \Cref{alg:path-following} outputs a vector $x>\vec{0}$ with $\norm{XAX\vec{1}-\vec{1}}_2\leq\eps$ in
	\[
	O\left(
	\T_{\SDD}\left(\nnz(A),\frac{\eps}{\poly(d,q)}\right)
	\left(\log(dq)+\log\log\frac{1}{\eps}\right)
	\right)
	\]
	time and
	\[
	O\left(
	\D_{\SDD}\left(\nnz(A),\frac{\eps}{\poly(d,q)}\right)
	\left(\log(dq)+\log\log\frac{1}{\eps}\right)
	\right)
	\]
	parallel depth, where
	$q\defeq\norm{A}_{\infty}\norm{x_0}_2^2$.
\end{theorem}
\begin{proof}
By \Cref{lem:initialization}, Line~\ref{line:initialize} gives the initial centrality $\Delta_{\mu}^\infty(x)=0 \leq 1/2$.
Thus, by \Cref{lem:approx-step-unified} with $\delta = 1/2$, each iteration of the first while loop halves $\mu$ and maintains
\[ \Delta_{\mu}^\infty(x) \leq \frac{9}{8}\left(\frac{1}{2} + \frac{1}{2} \cdot \frac{1/2}{2 - 1/2}\right)^2 \leq \frac{1}{2} \]
after Line~\ref{line:end-first-loop}.
This uses the observation that in \Cref{alg:path-following} we solve each $h'$ to accuracy $\eps/(10^5d^2q^{3}) \leq (1/2)/(10^4 d q^3)$ as required by \Cref{lem:approx-step-unified}.
Thus, until $\mu\leq \frac{1}{12 d\norm{b}_2\norm{x_0}_\infty}$, the first while loop performs $O(\log (d\norm{b}_2\norm{x_0}_{\infty}))$ iterations. 
Since
\[ b=X_0^{-1}(X_0AX_0\vec{1} -\vec{1})\leq X_0^{-1}(X_0AX_0\vec{1})=Ax_0,\]
we have $\norm{b}_2\leq \norm{Ax_0}_2\leq \norm{A}_2\norm{x_0}_2$ (note that $b \geq \vec{0}$).
Thus,  the number of iterations is  $O(\log(dq))$.

Note that since $X_0 A X_0 \vec{1} \geq \vec{1}$, we have $1/\norm{x_0}_{\infty} \leq \sqrt{\lambda_{\min}(A)}$.
Therefore, when the first loop ends, we have $\mu\leq \frac{\sqrt{\lambda_{\min}(A)}}{12\sqrt{d}\norm{b}_2}$. 
Together with $\Delta_{\mu}^\infty(x)\leq 1/2$, by \Cref{lem:mu-set-zero}, we have $\Delta^{\infty}_0(x) = \norm{XAX\vec{1}-\vec{1}}_\infty\leq 2/3$.
Hence, by \Cref{lem:approx-step-unified} with $\delta = 0$, each iteration of the second while loop decreases $\norm{XAX\vec{1}-\vec{1}}_\infty$ quadratically.
Again, \Cref{lem:approx-step-unified} applies since $\eps/(10^5d^2q^3) \leq (\eps/\sqrt{d})/(10^4dq^3)$, noting that we only run the loop when $\Delta^{\infty}_0(x) \geq \eps/\sqrt{d}$.
It takes $O(\log\log (d/\eps))$ iterations to make $\norm{XAX\vec{1} - \vec{1}}_{\infty} \leq \eps/\sqrt{d}$ which implies $\norm{XAX\vec{1} - \vec{1}}_2 \leq \eps$.
This establishes the correctness of \Cref{alg:path-following}.

As for the efficiency, \Cref{alg:path-following} performs a total of
$O(\log(dq) + \log\log(1/\eps))$
iterations.
By \Cref{lem:solver-cost-per-call}, each iteration runs in $O(\T_{\SDD}(\nnz(A), \eta/\kappa(X)^2))$ time and $O(\D_{\SDD}(\nnz(A), \eta/\kappa(X)^2))$ parallel depth for the current $x$, where $\eta \defeq \eps/(10^5d^2q^{3})$ and $q \defeq \norm{A}_{\infty}\norm{x_0}_2^2$ as in \Cref{alg:path-following}.
The theorem thus follows from \Cref{claim:X-entry} which shows that $\kappa(X) \leq 10q$.

\end{proof}

\subsection{Initializing the Path-Following Algorithm}\label{sub:predcondition}

We now explain how to compute the initial $x_0$ such that $X_0 A X_0 \vec{1} \geq \vec{1}$.
This is essentially done by running the algorithm of \cite{AhmadinejadJSS19} that, given a symmetric $M$-matrix $A$, computes a scaling $y \in \R^d_{>0}$ such that $YAY$ is SDD.
The following lemma is obtained by carefully analyzing the parameters of \cite{AhmadinejadJSS19} and is thus deferred to \Cref{appendix:omitted_proof}.

\begin{restatable}{lemma}{rcdd}\crefalias{theorem}{lemma}\label{lem:RCDD}
  There is an algorithm that,
  given a symmetric $M$-matrix $A \in \R^{d \times d}$, computes a
  $v \in \R^d_{>0}$ such that
	\[
	VAV\vec{1} \geq \frac{\norm{A}_\infty}{32}\vec{1}\quad\text{and}\quad
	\norm{v}_2\leq 3d\kappa(A)
	\]
	with high probability in
	\[
    O\left(
    \T_{\SDD}\left(\nnz(A),\frac{1}{\poly(d, \kappa(A))}\right)
    \log^2(d\kappa(A))
    \right)
	\]
    time
	and
	\[
    O\left(
    \D_{\SDD}\left(\nnz(A),\frac{1}{\poly(d,\kappa(A))}\right)
    \log^2(d\kappa(A))
    \right)
	\]
     parallel depth.
\end{restatable}

With \Cref{thm:path-following,lem:RCDD} we can now reduce symmetric $M$-matrix scaling to $\tilde{O}(1)$ iterations of SDD system solves.

\begin{theorem}\label{thm:alg-matrix-scaling}
There is an algorithm that, given a symmetric $M$-matrix $A \in \R^{d \times d}$ and a parameter $\eps \in (0, 1)$, computes an $x \in \R^d_{>0}$ such that $\|XAX\vec{1} - \vec{1}\|_2 \leq \eps$ with high probability
	in 
	\[O\left(\T_{\SDD}\left(\nnz(A), \frac{\eps}{\poly(d,\kappa(A))}\right) \cdot \left(\log^2(d\kappa(A)) + \log\log \frac{1}{\eps}\right)\right)
	\]
    time
	and
	\[O\left(\D_{\SDD}\left(\nnz(A), \frac{\eps}{\poly(d,\kappa(A))}\right) \cdot \left(\log^2(d\kappa(A)) + \log\log \frac{1}{\eps}\right)\right)
	\]
    parallel depth.
\end{theorem}

\begin{proof}
  We compute the vector $v \in \R^{d}_{>0}$ in \Cref{lem:RCDD} and set $x_0 \defeq \sqrt{32/\norm{A}_{\infty}} \cdot v$ which satisfies $X_0 A X_0 \vec{1} \geq \vec{1}$ and $\norm{x_0}_2^2 \leq \poly(d,\kappa(A)) / \norm{A}_{\infty}$.
  This takes
  \[ O\left(\T_{\SDD}\left(\nnz(A), \frac{1}{\poly(d,\kappa(A))}\right) \cdot \log^2(d\kappa(A))\right) \]
  time and
  \[ O\left(\D_{\SDD}\left(\nnz(A), \frac{1}{\poly(d,\kappa(A))}\right) \cdot \log^2(d\kappa(A))\right) \]
  parallel depth.
  By \Cref{thm:path-following}, with $q \defeq \norm{A}_{\infty}\norm{x_0}_2^2 \leq \poly(d,\kappa(A))$, starting  from $x_0$ we can compute an $\eps$-approximate scaling for $A$ in
  \[ O\left(\T_{\SDD}\left(\nnz(A), \frac{\eps}{\poly(d,\kappa(A))}\right) \cdot \left(\log(d\kappa(A)) + \log\log \frac{1}{\eps}\right)\right)\]
  time and
  \[ O\left(\D_{\SDD}\left(\nnz(A), \frac{\eps}{\poly(d,\kappa(A))}\right) \cdot \left(\log(d\kappa(A)) + \log\log \frac{1}{\eps}\right)\right)\]
  parallel depth.
  This concludes the proof.
\end{proof}

We can now prove \Cref{thm:main-randomized} by plugging in known SDD solvers.

\begin{proof}[Proof of \Cref{thm:main-randomized}]
  By \cite[Theorem 1.1]{SachdevaZ23} we have
  $\T_{\SDD}(m, \eta) = O(m\log^3 m \log\log m \log(1/\eta))$ and 
  $\D_{\SDD}(m, \eta) = O(\log^2 m \log\log m \log (1/\eta))$.
  The rest follows from \Cref{thm:alg-matrix-scaling}.
\end{proof}

%% file: path-following-algorithm2e.tex
\begin{algorithm2e}[h!]
	\caption{\textsc{PathFollowing}($A$, $x_0$, $\eps$)}\label{alg:path-following}
	\KwInput{A symmetric $M$-matrix $A\in \R^{d\times d}$, a vector $x_0\in \R^d_{>0}$ such that $X_0AX_0\vec{1}\geq\vec{1}$, and a precision parameter $\eps \in (0, 1)$.}
	\KwOut{A vector $x>\vec{0}$ satisfying $\norm{XAX\vec{1}-\vec{1}}_2\leq\eps$.}
    \SetKwFunction{solve}{Solve}
    
	$b\gets X_0^{-1}(X_0AX_0\vec{1}-\vec{1})$, $\mu\gets 1$, and $x\gets x_0$\; \label{line:initialize}

    \lIf{$b = \vec{0}$} { \Return{$x$} }
    \,

    \tcp{Let $q \defeq \norm{A}_{\infty}\norm{x_0}_2^2$ (note that $q \geq 1$) and $\eta \defeq \eps / (10^5 d^2 q^{3})$.}
    \tcp{The first while loop decreases $\mu$ to a sufficiently small value.}
	\While{$\mu> \frac{1}{12 d \norm{b}_2 \norm{x_0}_{\infty}}$}{
		$\mu\gets\mu/2$ \;\label{line:mu-halve}
       Let $h'$ be an arbitrary vector in $\A_{\eta}(\nabla^2 f_{\mu}(x), \nabla f_{\mu}(x))$\;
		$x\gets x-h'$\;\label{line:end-first-loop}
	}

    \tcp{Now that $\mu$ is small, we set it straight to $0$, and the second while loop directly decreases $\Delta_0^{\infty}(x) = \|XAX \vec{1} - \vec{1}\|_{\infty}$.}
	\While{$\norm{XAX\vec{1}-\vec{1}}_\infty>\frac{\eps}{\sqrt{d}}$}{
        Let $h'$ be an arbitrary vector in $\A_{\eta}(\nabla^2 f_0(x), \nabla f_0(x))$\;
		$x\gets x-h'$\;\label{line:end-second-loop}
	}
	\Return{$x$}
\end{algorithm2e}

%% file: ack.tex
\section*{Acknowledgments}

We thank Adrian Vladu for helpful conversations.
Part of this research took place while Aaron Sidford was visiting National Taiwan University.
Chih-Hung Liu was supported by Ministry of Education, Taiwan under Yushan Fellow Program with the grant number MOE-111-YSFEE-0003-006-P1 and by National Science and Technology Council, Taiwan with the grant numbers NSTC-113-2221-E-002-209-MY3 and NSTC-114-2628-E-002-014-MY3.

%% file: flow.tex
\section{Almost-Linear Time $M$-Matrix Scaling via Convex Flow} \label{appendix:convex-flow}

In this section we show how one can get an almost-linear time algorithm for symmetric $M$-matrix scaling using the convex flow algorithm of \cite{ChenKLPGS22}.
Similar to our algorithm, we may assume without loss of generality that $A$ is diagonally dominant by first running \cite{AhmadinejadJSS19} to obtain such a scaling.
We can write $A = C + D$ where $C$ is a Laplacian matrix (i.e., $C \vec{1} = \vec{0}$) and $D$ is a nonnegative diagonal matrix.
It is standard that $C = B^\top W B$ where $B$ is the edge-vertex incidence matrix of some graph $G$ and $W$ is a nonnegative diagonal matrix.
Let $m$ be the number of edges in $G$.

By the optimality condition, the solution to the matrix scaling problem is the minimizer of
\begin{equation}
  \min_{x > \vec{0}} \frac{1}{2}x^\top A x - \sum_{i \in [d]}\ln x_i = \min_{y \in \R^m, x > \vec{0}, y = B x} \frac{1}{2} y^\top W y + \frac{1}{2} x^\top D x - \sum_{i \in [d]}\ln x_i.
  \label{eq:primal}
\end{equation}

The dual of \eqref{eq:primal} is

\begin{equation}
  \min_{f \in \R^m, x > \vec{0}, Dx + B^\top f - 1/x \geq \vec{0}} \frac{1}{2}f^\top W^{-1} f + \frac{1}{2}x^\top D x + \sum_{i \in [d]} \ln x_i = \min_{f \in \R^m, b = B^\top f} \frac{1}{2}f^\top W^{-1} f + \sum_{i \in [d]}\phi_i(b_i),
  \label{eq:dual}
\end{equation}
where
\[ \phi_i(b_i) = \min_{x_i > 0, d_i x_i - 1/x_i \geq -b_i} \frac{1}{2} d_i x_i^2 + \ln x_i \]
is convex (and can be computed explicitly).
This becomes a convex flow problem.
Let us add one vertex $v_*$ to $G$ and edges $(v_*, v)$ for each original vertex $v$.
The capacity of each edge is infinite in both directions.
The cost of an original edge $e$ is $c_e(f_e) = \frac{1}{2} w_e^{-1} f_e^2$.
The cost of the edge $(v_*, v)$ is $c_e(f_e) = \phi_v(f_e)$.
The minimum-cost flow in this graph can be straightforwardly converted back to the minimizer of \eqref{eq:primal} via duality.
We omit the details on approximation for brevity.

%% file: omitted_proof.tex
\section{Omitted Proofs}\label{appendix:omitted_proof}

In this section we provide proofs that are omitted from the main body of the paper.
We start with \Cref{claim:X-entry} which is needed for proving \Cref{lem:approx-step-unified}.

{\renewcommand{\thetheorem}{\ref*{claim:X-entry}}
\Xentry*}

\begin{proof}
With $\Delta^{\infty}_{\mu}(x) \leq 2/3$ we see that
\begin{equation}
  \norm{X(Ax - \mu b) - \vec{1}}_{\infty} \leq 2/3 \implies 1/3 \leq 1/3 + \mu x_ib_i \leq x_i (Ax)_i \leq 1 + 2/3 + \mu x_i b_i.
  \label{eq:norm-bound}
\end{equation}
This establishes $XAx \geq 1/3 \cdot \vec{1}$.
We now prove $x \leq 5/3 \cdot x_0$.
Let $t \defeq \max_{j \in [d]}x_j / (x_0)_j$ and let $i \in [d]$ be such that $x_i = t(x_0)_i$.
The goal is to show that $t \leq 5/3$.
From \eqref{eq:norm-bound} we get
\begin{equation}
  x_i (Ax)_i \leq 5/3 + x_i ((Ax_0)_i - 1/(x_0)_i) \implies t = x_i / (x_0)_i \leq 5/3 + x_i((Ax_0)_i - (Ax)_i).
  \label{eq:t}
\end{equation}
Observe that
\begin{equation}
  (Ax)_i = A_{ii}x_i + \sum_{j \neq i}A_{ij}x_j \geq t \cdot (Ax_0)_i
  \label{eq:Ax}
\end{equation}
because each $A_{ij} \leq 0$ for $i\neq j$.
Now either $t \leq 1$ already, or $t > 1$ in which case \eqref{eq:Ax} suggests $(Ax)_i > (Ax_0)_i$ and therefore \eqref{eq:t} implies $t \leq 5/3$.

To upper bound $\kappa(X)$, we use the lower bound in \eqref{eq:norm-bound}.
We notice that $(Ax)_i \geq 0$ and in particular $0 \leq (Ax)_i \leq \norm{A}_{\infty}\norm{x}_{\infty}$ gives
\[
  x_i \geq \frac{1}{3\norm{A}_{\infty}\norm{x}_{\infty}} \implies \kappa(X) \leq 3\norm{A}_{\infty}\norm{x}_{\infty}^2 \leq 10q,
\]
where we used $\norm{x}_{\infty} \leq t \norm{x_0}_{\infty}$ in the last inequality.
\end{proof}

{\renewcommand{\thetheorem}{\ref*{lem:approx-step-unified}}
\approxstepunified*}
\begin{proof}
    Let us abbreviate $r \defeq \delta + (1-\delta)\Delta^{\infty}_{\mu}(x)/(2-\Delta^{\infty}_{\mu}(x))$.
    Let $\mu'\defeq (1-\delta)\mu$ and $y \defeq x - h$ where $h \defeq (\nabla^2 f_{\mu'}(x))^{-1} \nabla f_{\mu'}(x)$ is the exact Newton step.
    Let us first bound the difference between $h$ and $h'$.
    By definition, we have
    \begin{equation}
      \norm{h - h'}_2 \leq \eta \norm{(\nabla^2 f_{\mu'}(x))^{-1}}_2\norm{\nabla f_{\mu'}(x)}_2 \leq \eta \norm{(\nabla^2 f_{\mu'}(x))^{-1}}_2\norm{\nabla^2 f_{\mu'}(x)}_2\norm{h}_2
    \end{equation}
    using that $\nabla f_{\mu'}(x) = \nabla^2 f_{\mu'}(x) h$.
    We should thus bound $\norm{(\nabla^2 f_{\mu'}(x))^{-1}}_2\norm{\nabla^2 f_{\mu'}(x)}_2 = \kappa(\nabla^2 f_{\mu'}(x))$.
    By \Cref{claim:X-entry}, we have $\min_{i \in [d]}x_i^2 \geq 1/(3\norm{A}_{\infty})$.
    Since $\nabla^2 f_{\mu'}(x) = A + X^{-2}$ we have $\lambda_{\max}(\nabla^2 f_{\mu'}(x)) \leq 4\norm{A}_{\infty}$.
    On the other hand, $\lambda_{\min}(\nabla^2 f_{\mu'}(x)) \geq \lambda_{\min}(A) \geq 1/\norm{x_0}^{2}_{\infty}$ since $X_0AX_0\vec{1} \geq \vec{1}$.
    Together this suggests that $\kappa(\nabla^2 f_{\mu'}(x)) \leq 4q$.

    This allows us to bound
    \[
      \norm{X^{-1}(y' - y)}_{\infty} \leq \norm{X^{-1}}_{\infty}\norm{h' - h}_2 \leq 4\eta q \cdot \norm{X^{-1}}_{\infty} \norm{h}_2 \leq 4d\eta qr \cdot \kappa(X) \leq 40d\eta q^2r,
    \]
    where we used that $\norm{X^{-1}h}_{\infty} \leq r$ from \Cref{lem:mu-update} which implies $\norm{h}_2 \leq d\norm{h}_{\infty} \leq \norm{X}_{\infty} dr$.
    We can now show that $y' > \vec{0}$.
    Indeed, since $X^{-1}y \geq (1 - \norm{X^{-1}h}_{\infty})\vec{1} \geq (1-r)\vec{1}$, we have
    \begin{equation}
      X^{-1}y' \geq X^{-1}y - 40d\eta q^2 r \cdot \vec{1} \geq (1-r - 40d\eta q^2 r)\vec{1} > \vec{0},
      \label{eq:y-lb}
    \end{equation}
    noting that $r \leq 3/4$ and applying the bound on $\eta$.
    It remains to bound $\Delta^{\infty}_{\mu'}(y')$.
    Similar to \eqref{eq:y-lb}, we have $\norm{X^{-1}y'}_{\infty} \leq 1+r+40d\eta q^2 r \leq 2$.
    We can write
    \begin{equation}
      Y'\nabla f_{\mu'}(y') = Y\nabla f_{\mu'}(y) + Y'A(y'-y) + (Y'-Y)(Ay - \mu'b).
      \label{eq:grad-y}
    \end{equation}
    Let us bound the two extra terms in \eqref{eq:grad-y}.
    We have
    \[
      \norm{Y'A(y'-y)}_{\infty} \leq \norm{Y'X^{-1}}_{\infty}\norm{XAX}_{\infty}\norm{X^{-1}(y'-y)}_{\infty} \leq 80d\eta q^2 r \norm{XAX}_{\infty} \leq 240d\eta q^3 r,
    \]
    where we used $\norm{XAX}_{\infty} \leq \norm{X}_{\infty}^2\norm{A}_{\infty} \leq 3q$ since $x \leq (5/3)x_0$ from \Cref{claim:X-entry},
    and
    \[
      \norm{(Y'-Y)(Ay - \mu'b)}_{\infty} \leq \norm{X^{-1}(y'-y)}_{\infty}\norm{X(Ay - \mu'b)}_{\infty} \leq 80d\eta q^2 r,
    \]
    where we used that
    \[
      Ay - \mu'b = Ay - (Ax - X^{-1}\vec{1} - \nabla f_{\mu'}(x)) = -Ah + (A + X^{-2}) h + X^{-1}\vec{1}
    \]
    and therefore $\norm{X(Ay - \mu'b)}_{\infty} \leq 1 + \norm{X^{-1}h}_{\infty} \leq 2$.
    Overall, \eqref{eq:grad-y} together with the observation that $\Delta^{\infty}_{\mu}(x) + \delta \leq 2r$ suggests
    \[
      \Delta^{\infty}_{\mu'}(y') \leq \Delta^{\infty}_{\mu'}(y) + 320d\eta q^3 r \leq \frac{9}{8}r^2.\qedhere
    \]
\end{proof}

{\renewcommand{\thetheorem}{\ref*{lem:solver-cost-per-call}}
\solvercostpercall*}

\begin{proof}
First, recall that $\nabla^2 f_{\mu'}(x) = X^{-1} (I + A(x)) X^{-1}$.
Note that $\nabla^2 f_{\mu'}(x)$ itself may not be an SDD matrix.
Nevertheless, we claim that $I + A(x)$ is.
Indeed, by \Cref{lem:inverse-bound-from-delta} (in particular \eqref{eq:matrix-bound}, since $\Delta^{\infty}_{\mu}(x) \leq 2/3$ for some $\mu \geq 0$, we have $(I + A(x))\vec{1} \geq \vec{0}$ which together with the fact that $I + A(x)$ is a symmetric $Z$-matrix shows that $I + A(x)$ is SDD.
Therefore, to compute $z \defeq (\nabla^2 f_{\mu}(x))^{-1} v$ for a given $v$, we can compute
\begin{equation}
  u \gets Xv,\quad w \gets (I + A(x))^{-1} u, \quad\text{and then}\quad z \gets Xw.
  \label{eq:algo}
\end{equation}

It remains to analyze the accuracy required for the linear system solve in $I + A(x)$ above.
We claim that if $w' \in \A_{\eta/\kappa(X)^2}(I + A(x), u)$, then $z' \defeq Xw'$ will be in $\A_{\eta}(\nabla^2 f_{\mu}(x), v)$ which would prove the lemma by plugging in $v = \nabla f_{\mu'}(x)$.
The claim follows from direct calculation
\begin{align*}
  \norm{w' - w}_2 \leq \frac{\eta}{\kappa(X)^2}\norm{(I + A(x))^{-1}}_2\norm{u}_2
  &\implies
  \norm{z' - z}_2 \leq \frac{\eta}{\kappa(X)^2}\norm{X}_2\norm{(I + A(x))^{-1}}_2\norm{u}_2 \\
  &\implies
  \norm{z' - z}_2 \leq \frac{\eta}{\kappa(X)^2}\norm{X}^2_2\norm{(I + A(x))^{-1}}_2\norm{v}_2 \\
  &\implies
  \norm{z' - z}_2 \leq \eta\norm{(\nabla^2 f_{\mu'}(x))^{-1}}_2\norm{v}_2,
\end{align*}
where we used that $\norm{(I + A(x))^{-1}}_2 \leq \norm{X^{-1}}_2^2\norm{(\nabla^2 f_{\mu'}(x))^{-1}}_2$ and that $\norm{X}_2\norm{X^{-1}}_2 = \kappa(X)$.
\end{proof}

{\renewcommand{\thetheorem}{\ref*{lem:RCDD}}
\rcdd*}

\begin{proof}

    Let $s \defeq 2\norm{A}_{\infty}$.
    Since $s > \norm{A}_2$ we have for $B \defeq sI - A$ that $\rho(B/s) < 1$.
    We run \cite[Algorithm 10]{AhmadinejadJSS19} (which underlies \cite[Theorem 8]{AhmadinejadJSS19}) on $B/s$.
    The algorithm starts from $\alpha = 1$ and repeatedly halves $\alpha$ while maintaining a vector $v_{\alpha}$ such that $\norm{M_{\alpha}v_{\alpha} - \vec{1}}_{\infty} \leq 1/2$ for $M_{\alpha} = (1+\alpha)I - B/s = A/s + \alpha I$.
    We run the algorithm until $\alpha\norm{v_{\alpha}}_{\infty} \leq 1/4$.\footnote{The actual \cite[Algorithm 10]{AhmadinejadJSS19} has an input accuracy $\eps$ and it runs the loop until $\alpha \leq \eps$.}

    Let us analyze the vector $v_{\alpha}$ when $\alpha\norm{v_{\alpha}}_{\infty} \leq 1/4$.
    We have
    \[
      \left(\frac{A}{s} + \alpha I\right)v_{\alpha} \geq \frac{1}{2}\vec{1} \implies V_{\alpha}AV_{\alpha}\vec{1} \geq s \cdot \left(\frac{v_{\alpha}}{2} - \alpha V_{\alpha}v_{\alpha}\right).
    \]
    Observe that since $A_{ij} \leq 0$ for $i \neq j$ and $A_{ii} \leq s$ for all $i \in [d]$, letting $z \defeq \argmin_{i\in [d]} [v_{\alpha}]_i$ we have
    \[
      \left[\left(\frac{A}{s} + \alpha I\right)v_{\alpha}\right]_z \leq (1+\alpha)[v_\alpha]_z
    \]
    which means that $v_{\alpha} \geq 1/(2(1+\alpha))\cdot \vec{1} \geq 1/4\cdot\vec{1}$.
    Therefore, it holds that $V_{\alpha} A V_{\alpha}\vec{1} \geq \norm{A}_{\infty}/32 \cdot \vec{1}$ as required (note again that $\alpha\norm{v_{\alpha}}_{\infty} \leq 1/4$).

    To bound $\norm{v}_2$, observe that $(A/s + \alpha I)$ is an $M$-matrix.
    Therefore, $(A/s + \alpha I)^{-1}$ is nonnegative and consequently
    \begin{equation}
      v_{\alpha} \leq \frac{3}{2}\left(\frac{A}{s} + \alpha I\right)^{-1}\vec{1} \leq \frac{3s}{2}A^{-1}\vec{1} \implies \norm{v_{\alpha}}_2 \leq 3\sqrt{d} \cdot \norm{A}_{\infty}\norm{A^{-1}}_2 \leq 3d \kappa(A).
      \label{eq:v-bound}
    \end{equation}

    As for the efficiency, note that \eqref{eq:v-bound} implies that it takes $O(\log(d\kappa(A)))$ iterations until $\alpha \norm{v_{\alpha}}_{\infty} \leq 1/4$.
    Each iteration, by \cite[Proof of Theorem 8]{AhmadinejadJSS19}, runs in $\T_{\SDD}(\nnz(A),1/\poly(d,\kappa(A))) \cdot O(\log(d\kappa(A)))$ time and $\D_{\SDD}(\nnz(A),1/\poly(d,\kappa(A))) \cdot O(\log(d\kappa(A)))$ parallel depth.\footnote{\cite{AhmadinejadJSS19} used a different notion of accuracy compared to our $\T_{\SDD}$ notation. This is why the $\T_{\mathrm{solve}}(n,m,1/4)$ term in \cite{AhmadinejadJSS19} translates to $\T_{\SDD}(\nnz(A),1/\poly(d,\kappa(A)))$.}
    This completes the proof.
\end{proof}